\documentclass[11pt,onecolumn,twoside]{IEEEtran} 
\usepackage{graphicx}
\usepackage{amsmath}
\usepackage{mathrsfs}
\usepackage{mathtools}
\usepackage{amssymb}
\usepackage{float}
\usepackage{url}
\usepackage{verbatim}
\usepackage{amsthm}
\usepackage{enumerate}
\usepackage{layout}
\usepackage{bm}

\newtheorem{theorem}{Theorem}
\newtheorem{lemma}{Lemma}
\newtheorem{proposition}{Proposition}

\newtheorem{definition}{Definition}
\theoremstyle{definition}

\newtheorem{remark}{Remark}
\newtheorem{problem}{Problem}

\newcommand{\naturals}{\ensuremath{\mathbb{N}}}
\newcommand{\Reals}{\ensuremath{\mathbb{R}}}
\newcommand{\expectation}{\ensuremath{\mathbb{E}}}
\newcommand{\Var}{\mathrm{Var}}

\newcommand{\set}{\ensuremath{\mathcal}}

\begin{document}
\thispagestyle{empty}
\setcounter{page}{1}
\setlength{\baselineskip}{1.15\baselineskip}

\title{\huge{A Tight Lower Bound for the Hellinger Distance with Given Means and Variances}\\[0.2cm]}
\author{Tomohiro Nishiyama\\ Email: htam0ybboh@gmail.com}
\date{}
\maketitle
\thispagestyle{empty}

\begin{abstract}
The binary divergences that are divergences between probability measures defined on the same 2-point set have an interesting property. For the chi-squared divergence and the relative entropy, it is known that their binary divergence attain lower bounds with given means and variances, respectively. In this note, we show that the binary divergence of the squared Hellinger distance has the same property and propose an open problem that what conditions are needed for f-divergence to satisfy this property.
\end{abstract}
 
\section{Introduction}
The Hellinger distance~\cite{hellinger1909neue} is a divergence measure which plays a key role in information theory, statistics and machine learning, and other fields in mathematics.
This belongs to an important class of divergence measures defined by means of convex functions $f$, and named $f$-divergences~\cite{csiszar1967information,csiszar1967topological,csiszar1972class}.
This class unifies other useful divergence measures such as the relative entropy (also known as the Kullback-Leibler divergence~\cite{kullback1951information}), and the chi-squared divergence~\cite{pearson1900x}. Regarding the chi-squared divergence, a tight lower bound with given means and variances is known as the Hammersley-Chapman-Robbins bound~\cite{chapman1951minimum}. Recently, we derive a tight lower bound under the same constraints for the relative entropy by using an integral relation between the relative entropy and the chi-squared divergence~\cite{nishiyama2020relations, nishiyama2019new}. Interestingly, both lower bounds are attained by their binary divergences that are divergences between probability measures on the same $2$-point set. On the other hand, the lower bound for the squared Hellinger distance given means and variances is studied in~\cite{dashti2013bayesian, katsoulakis2017scalable}.

In this note, we derive a tight lower bound for the squared Hellinger distance, and show its binary divergence attains the lower bound as in the case of the chi-squared divergence and the relative entropy. Furthermore, we propose an open problem that what conditions are needed for $f$-divergence to satisfy this property.

\section{Lower bound for the squared Hellinger distance}
\subsection{Preliminaries}

This subsection provides definitions of divergence measures which are used in this note.
\begin{definition} {\rm \label{def:fD} \cite[p.~4398]{liese2006divergences}}
Let $P$ and $Q$ be probability measures, let $\mu$ be a dominating measure
of $P$ and $Q$ (i.e., $P, Q \ll \mu$), and let $p := \frac{\mathrm{d}P}{\mathrm{d}\mu}$
and $q := \frac{\mathrm{d}Q}{\mathrm{d}\mu}$ be the densities of $P$ and $Q$ with respect
to $\mu$. The {\em $f$-divergence} from $P$ to $Q$ is given by
\begin{align} \label{eq:fD}
D_f(P\|Q) := \int q \, f \Bigl(\frac{p}{q}\Bigr) \, \mathrm{d}\mu,
\end{align}
where
\begin{align}
& f(0) := \underset{t \to 0^+}{\lim} \, f(t), \quad  0 f\biggl(\frac{0}{0}\biggr) := 0, \\[0.1cm]
& 0 f\biggl(\frac{a}{0}\biggr)
:= \lim_{t \to 0^+} \, t f\biggl(\frac{a}{t}\biggr)
= a \lim_{u \to \infty} \frac{f(u)}{u}, \quad a>0.
\end{align}
It should be noted that the right side of \eqref{eq:fD} is invariant
in the dominating measure $\mu$.
\end{definition}
\begin{definition} \label{def:Hellinger-distance}
The {\em squared Hellinger distance} is the $f$-divergence with $f(t) := \frac12 (\sqrt{t}-1)^2$ or $1-\sqrt{t}$ 
for $t >0$,
\begin{align}
\label{eq-Hellinger}
H^2(P, Q) &:= D_f(P\|Q) \\
&=\frac12 \int (\sqrt{p}-\sqrt{q})^2 \mathrm{d}\mu.
\end{align}
\end{definition}
The relative entropy and the chi-square divergence are the $f$-divergence with $f(t):=t\log t$ and $f(t):=(t-1)^2$, respectively.
\begin{definition} \label{def: probability_set}
Let us define a set of pairs of probability measures $(P,Q)$ defined on $n$-point set $\{u_1, u_2, \cdots, u_n\}$ by $\set{P}_n$, where $\{u_i\}_{1\leq i\leq n}$ are arbitrary real numbers.
\end{definition}
If $m< n$, $\set{P}_m$ is a subset of $\set{P}_n$.
\begin{definition} \label{def:binary Hellinger-distance}
The {\em binary squared Hellinger distance} is defined for $(P,Q)\in \set{P}_2$.
\begin{align}
\label{eq-binary Hellinger}
h^2(r,s) := \frac12 \Bigl((\sqrt{r}-\sqrt{s})^2+(\sqrt{1-r}-\sqrt{1-s})^2\Bigr),
\end{align}
where $P(u_1) = r$ and $Q(u_1) = s$.
\end{definition}
\begin{definition}
Let $P$ and $Q$ be probability measures defined on a measurable space $(\Reals, \mathscr{B})$, where $\Reals$ is the real line and $\mathscr{B}$ is the Borel $\sigma$-algebra of subsets of 
$\Reals$.
Let $\set{P}[m_P, \sigma_P, m_Q, \sigma_Q]$ be a set of pairs of probability measures $(P,Q)$ with given means and variances, i.e.,
\begin{align}
\label{constraints}
& \expectation[X] =: m_P, \; \expectation[Y] =: m_Q,
\quad \Var(X) =: \sigma_P^2, \;  \Var(Y) =: \sigma_Q^2,
\end{align}
where $X\sim P$ and $Y\sim Q$.
\end{definition}
%
\subsection{Main results}
\begin{theorem} \label{theorem_LB_Hellinger}
Let $(P, Q)\in \set{P}[m_P, \sigma_P, m_Q, \sigma_Q]$.
\begin{enumerate}[a)]
\item
If $m_P\neq m_Q$, then
\begin{align} \label{eq_LB_Hellinger}
H^2(P,Q) \geq h^2(r,s)=1-\sqrt{1-\frac{a^2}{a^2+(\sigma_P+\sigma_Q)^2}}, 
\end{align}
where 
\begin{align}
\label{r}
& r := \frac12 + \frac{b+a^2}{4av} \in [0,1], \\
\label{s}
& s :=  \frac12 + \frac{b-a^2}{4av}\in [0,1], \\
\label{a}
& a:= m_P - m_Q, \\
\label{b}
& b:= \sigma_Q^2-\sigma_P^2, \\
\label{v}
& v:= \frac{1}{2|a|}\sqrt{b^2 + 2a^2(\sigma_P^2+\sigma_Q^2)+a^4}.
\end{align}

\item
The lower bound in the right side of \eqref{eq_LB_Hellinger} is attained for $(P,Q)\in \set{P}_2$ defined on $\{u_1, u_2\}$, and
\begin{align} \label{eq_binary_prob}
P(u_1) = r, \quad Q(u_1) = s,
\end{align}
with $r$ and $s$ in \eqref{r} and \eqref{s}, respectively, and
\begin{align} \label{vec_u_1,2}
& u_1 :=  m_P + \sqrt{\frac{(1-r) \sigma_P^2}{r}},
\quad u_2 := m_P - \sqrt{\frac{r \sigma_P^2}{1-r}}.
\end{align}

\item
If $m_P=m_Q$, then,
\begin{align} \label{mean_equal}
\inf_{(P,Q)\in\set{P}[m_P,\sigma_P, m_Q, \sigma_Q]} H^2(P,Q)=0.
\end{align}
\end{enumerate}
\end{theorem}
\begin{proof}
See Subsection~\ref{subsection: proofs}.
\end{proof}
\begin{remark}
The Bhattacharyya coefficient~\cite{bhattacharyya1943measure} between $P$ and $Q$ is given by $\rho(P,Q):= \int \sqrt{pq} \mathrm{d}\mu$, then Theorem~\ref{theorem_LB_Hellinger} gives  the tight upper bound of $\rho(P,Q)$ for $P, Q\in \set{P}[m_P, \sigma_P, m_Q, \sigma_Q]$.
\end{remark}
We compare the binary squared Hellinger distance $h^2(r,s)$ and the lower bound shown in~\cite{katsoulakis2017scalable}, which is given by
\begin{align}
\label{eq_hel_lB1}
l(m_P,\sigma_P, m_Q, \sigma_Q) := \frac{a^2}{2\Bigl(a^2+2(\sigma_P^2 + \sigma_Q^2)\Bigr) },
\end{align} 
where $a$ is given by \eqref{a}. See Lemma A.1 of~\cite{katsoulakis2017scalable} for details on deriving \eqref{eq_hel_lB1}.
\begin{remark}
Note that the definition of the squared Hellinger distance differs from the definition in~\cite{katsoulakis2017scalable} by a factor $2$. 
\end{remark} 
\begin{proposition}  \label{proposition: hellinger_LB}
\label{prop_hellinger}
Let $g(x) := \frac{(1+x)}{(1+\sqrt{x})^2}$ for $x\in[0, \infty)$.
Then, 
\begin{align} 
\label{eq_compare}
\beta_{\mathrm{min}} l(m_P,\sigma_P, m_Q, \sigma_Q) \leq h^2(r,s) \leq \beta_{\mathrm{max}} l(m_P,\sigma_P, m_Q, \sigma_Q),
\end{align}
where $\beta_{\mathrm{min}}$ and $\beta_{\mathrm{max}}$ are given by
\begin{align}
\beta_{\mathrm{max}} := 2\max\Bigl(g\Bigl(\frac{s}{r}\Bigr), g\Bigl(\frac{1-s}{1-r}\Bigr)\Bigr) \in [1, 2], \\
\beta_{\mathrm{min}} := 2\min\Bigl(g\Bigl(\frac{s}{r}\Bigr), g\Bigl(\frac{1-s}{1-r}\Bigr)\Bigr) \in [1, 2].
\end{align}
\end{proposition}
\begin{proof}
See Subsection~\ref{subsection: proofs}.
\end{proof}
From Proposition~\ref{proposition: hellinger_LB}, it can be seen that the difference between the two lower bounds increases when $|r-s|$ is large.
\subsection{Examples}

We compare numerically the lower bounds on the squared Hellinger distance, as it given in \eqref{eq_LB_Hellinger} and \eqref{eq_hel_lB1}, with the squared Hellinger distance for the Gaussian distribution and the exponential distribution:
\begin{enumerate}[a)]
\item
The squared Hellinger distance between real-valued Gaussian distribution is given by
\begin{align}
\label{eq_Gaussian}
H^2\Bigl(\mathcal{N}(m_P, \sigma_P^2), \mathcal{N}(m_Q, \sigma_Q^2)\Bigr)=1-\sqrt{\frac{2\sigma_P\sigma_Q}{\sigma_P^2+\sigma_Q^2}}\exp\Bigl(-\frac{(m_P-m_Q)^2}{4(\sigma_P^2+\sigma_Q^2)}\Bigr).
\end{align}
\item
Let $E_\mu$ denote a random variable which is exponentially distributed with mean $\mu> 0$; its probability density function is given by
\begin{align}
e_\mu(x)=\frac{1}{\mu} \exp\Bigl(-\frac{x}{\mu}\Bigr)\{x\geq 0\}.
\end{align}
Then, for $a_1, a_2 > 0$ and $d_1, d_2\in\Reals$, 
\begin{align}
\label{eq_exponential}
H^2(E_{a_1}+d_1, E_{a_2}+d_2)=
\begin{dcases}
1-\frac{2\sqrt{a_1a_2}}{a_1+a_2}\exp\Bigl({-\frac{(d_1-d_2)}{2a_2}}\Bigr), & \quad d_1\geq d_2, \\
1-\frac{2\sqrt{a_1a_2}}{a_1+a_2}\exp\Bigl({-\frac{(d_2-d_1)}{2a_1}}\Bigr), & \quad d_1< d_2.
\end{dcases}
\end{align}
In this case, the means under $P$ and $Q$ are $m_P=d_1+a_1$ and $m_Q=d_2+a_2$, respectively, and the variances are $\sigma_P^2=a_1$ and $\sigma_Q^2=a_2$.
Hence, for obtaining the required means and variances, set
\begin{align}
a_1=\sigma_P, \quad a_2=\sigma_Q, \quad d_1=m_P-\sigma_P, \quad d_2=m_Q-\sigma_Q.
\end{align}
\end{enumerate}
\begin{enumerate}[1)]
\item
If $(m_P, \sigma_P^2, m_Q, \sigma_Q^2) =(10, 100, 3, 9)$, then the two lower bounds in \eqref{eq_LB_Hellinger} and \eqref{eq_hel_lB1} are equal to $0.120$ and $0.092$, respectively.
Additionally, the two squared Hellinger distance in \eqref{eq_Gaussian} and \eqref{eq_exponential} are equal to $0.337$ and $0.157$, respectively.
\item
If $(m_P, \sigma_P^2, m_Q, \sigma_Q^2) =(20, 30, 10, 20)$, then the two lower bounds in \eqref{eq_LB_Hellinger} and \eqref{eq_hel_lB1} are equal to $0.295$ and $0.250$, respectively.
Additionally, the two squared Hellinger distance in \eqref{eq_Gaussian} and \eqref{eq_exponential} are equal to $0.400$ and $0.636$, respectively.
\end{enumerate}
\subsection{Proofs}
\label{subsection: proofs}
We first show the following three lemmas.
\begin{lemma}
\label{lem_p2}
A set $\set{P}_2 \cap \set{P}[m_P, \sigma_P, m_Q, \sigma_Q]$ has one component $(P,Q)$ that is given by \eqref{eq_binary_prob}, and \eqref{vec_u_1,2}.
\end{lemma}
\begin{proof}
We outline the calculation, and see the proof of Theorem 2 in~\cite{nishiyama2020relations} for detailed calculation. Let $(P,Q)$ be defined on $\{u'_1, u'_2\}$, and let $P(u'_1) = r'$ and $Q(u'_1) = s'$. From the mean and variance constraints for $P$, it follows that
\begin{align}
\label{eq_a1}
& u'_1 :=  m_P \pm \sqrt{\frac{(1-r') \sigma_P^2}{r'}},
\quad u'_2 = m_P \mp \sqrt{\frac{r' \sigma_P^2}{1-r'}}.
\end{align}
By subtracting the mean constraint for $P$ from that for $Q$, we have
\begin{align}
\label{eq_a2}
(s'-r')(u'_1-u'_2)=-a.
\end{align}
By subtracting the variance constraint for $P$ from that for $Q$, substituting \eqref{eq_a2} into the result of subtraction, and rearranging terms, we have  
\begin{align}
\label{eq_a3}
u'_1+u'_2=2m_P-a-\frac{b}{a}.
\end{align}
Substituting \eqref{eq_a1} into \eqref{eq_a3}, it follows that
\begin{align}
\label{eq_a4}
\pm \frac{\sigma_P(1-2r')}{\sqrt{r'(1-r')}}=-\frac{b+a^2}{a}.
\end{align}
Squaring both sides and rearranging terms, we have
\begin{align}
\label{eq_a5}
\sqrt{r'(1-r')}=\frac{\sigma_P}{2v}.
\end{align}
Substituting \eqref{eq_a5} into \eqref{eq_a4}, we obtain
\begin{align}
\label{eq_a6}
r'=\frac12 \pm\frac{b+a^2}{4av}.
\end{align}
Substituting \eqref{eq_a1}, \eqref{eq_a5}, and \eqref{eq_a6} into \eqref{eq_a2}, we finally obtain
\begin{align}
\label{eq_a7}
s'=r'-\frac{a}{2v}=\frac12 \pm\frac{b-a^2}{4av}.
\end{align}
One solution of \eqref{eq_a1}, \eqref{eq_a6}, and \eqref{eq_a7} gives \eqref{r}, \eqref{s}, and \eqref{vec_u_1,2}, and another solution gives the same probability measure by replacing $(r', s', u'_1, u'_2)\rightarrow (1-r, 1-s, u_2, u_1)$. 
\end{proof}

\begin{lemma}
\label{lem_binary}
If $(P,Q) \in \set{P}_2 \cap \set{P}[m_P, \sigma_P, m_Q, \sigma_Q]$, 
\begin{align}
\label{eq_binary}
H^2(P,Q)=h^2(r,s)=1-\sqrt{1-\frac{a^2}{a^2+(\sigma_P+\sigma_Q)^2}}.
\end{align}
\end{lemma}

\begin{proof}
From Lemma \ref{lem_p2}, $(P,Q)\in\set{P}_2 \cap \set{P}[m_P, \sigma_P, m_Q, \sigma_Q]$ is unique.
From\eqref{r} and \eqref{s}, we have
\begin{align}
\label{eq_b1}
\sqrt{s(1-s)}=\frac{\sigma_P}{2v}, \quad \sqrt{s(1-s)}=\frac{\sigma_Q}{2v}, \\[0.1cm]
\label{eq_b2}
\Bigl(r-\frac{1}{2}\Bigr)\Bigl(s-\frac{1}{2}\Bigr)=\frac{b^2-a^4}{16a^2v^2}.
\end{align}
By combining these relations and
\begin{align}
\label{eq_b3}
(\sqrt{rs}+\sqrt{(1-r)(1-s)})^2&=2\Bigl(r-\frac{1}{2}\Bigr)\Bigl(s-\frac{1}{2}\Bigr)+2\sqrt{r(1-r)}\sqrt{s(1-s)}+\frac{1}{2},
\end{align}
it follows that
\begin{align}
\label{eq_b4}
&(\sqrt{rs}+\sqrt{(1-r)(1-s)})^2=\frac{b^2+2a^2\sigma_P\sigma_Q-a^4}{8a^2v^2}+\frac{1}{2}\\[0.1cm]
\label{eq_b5}
&=\frac{b^2+2a^2(\sigma_P^2+\sigma_Q^2)+a^4}{8a^2v^2} - a^2\frac{(\sigma_P-\sigma_Q)^2+a^2}{4a^2v^2} +\frac{1}{2} \\[0.1cm]
\label{eq_b6}
&=1-\frac{a^2}{a^2+(\sigma_P+\sigma_Q)^2},
\end{align}
where we use \eqref{v} and $4a^2v^2=\Bigl(a^2+(\sigma_P+\sigma_Q)^2\Bigr)\Bigl(a^2+(\sigma_P-\sigma_Q)^2\Bigr)$.
By substituting \eqref{eq_b6} into $h^2(r,s)=\frac12 \Bigl((\sqrt{r}-\sqrt{s})^2+(\sqrt{1-r}-\sqrt{1-s})^2\Bigr)=1-(\sqrt{rs}+\sqrt{(1-r)(1-s)})$, we obtain \eqref{eq_binary}.
\end{proof}

\begin{lemma}
\label{lem_minimum}
For $R>0$, let $\set{P}_{n, R}\subseteq \set{P}_{n}$ be a set of pairs of probability measures  that satisfy $|u_i|\leq R$ for $i=1,2,\cdots, n$.
Let $(P,Q)\in \set{P}_{n, R}\cap \set{P}[m_P, \sigma_P, m_Q, \sigma_Q]$.

If $m_P\neq m_Q$, the global minimum points $(P^*, Q^*, \bm{u}^*)=\mathrm{argmin}_{(P,Q)\in\set{P}_{n,R}\cap \set{P}[m_P, \sigma_P, m_Q, \sigma_Q]} H^2(P,Q)$ satisfy any of the following conditions.
\begin{enumerate}[1)]
\item
$(P^*,Q^*)\in\set{P}_2$ and $\max_i |u^*_i| < R$.
\item
$\max_i |u^*_i| = R$.
\end{enumerate}
\end{lemma}
Note that $H^2(P,Q)$ depends on $\{u_i\}$.

\begin{proof}
We consider the case of $n\geq 3$ ($n=1,2$ are trivial), and let $p_i:= P(u_i)$, $q_i:= Q(u_i)$, and $p_i=w_i^2$, $q_i=z_i^2$. Consider the following minimization problem.
\begin{align}
\label{eq_min_problem}
&\mbox{minimize} \quad - \sum_i w_i z_i, \\[0.1cm]
\label{eq_constraint1}
\quad \mbox{subject to}\quad &g_k(\bm{w}, \bm{u}):= \sum_i w_i^2 u_i^{k-1} -A_k=0, \\[0.1cm] 
\label{eq_constraint2}
& g_{k+3}(\bm{z}, \bm{u}):= \sum_i z_i^2 u_i^{k-1} -B_k=0, \quad \mbox{for} \hspace*{0.15cm}  k=1,2,3,\\[0.1cm]
\label{eq_constraint3}
&|w_i| \leq 1, \quad |z_i| \leq 1, \quad |u_i|\leq R, \quad \mbox{for} \hspace*{0.15cm}  1\leq i\leq n,
\end{align}
where $\bm{A}:= (1, m_P, \sigma_P^2 + m_P^2)^\mathrm{T}$, $\bm{B}:= (1, m_Q, \sigma_Q^2 + m_Q^2)^\mathrm{T}$, and \eqref{eq_constraint1} and \eqref{eq_constraint2} correspond to \eqref{constraints}.
Since the feasible set is compact, there exists a global minimum. If $(\bm{w}^*, \bm{z}^*, \bm{u}^*)$ is the global minimum point, the point $\Bigl((|w^*_1|, |w^*_2|, \cdots, |w^*_n|), (|z^*_1|, |z^*_2|, \cdots, |z^*_n|), \bm{u}^*\Bigr)$ is also the global minimum point because $-\sum_i w^*_i z^*_i\geq -\sum_i |w^*_i| |z^*_i|$.
Hence, the global minimum of the problem \eqref{eq_min_problem}-\eqref{eq_constraint3} is equal to the global minimum of the original problem, and global minimum points in the area of non-negative $\{w_i\}$ and $\{z_i\}$ are the same as original ones. The advantage of replacing the problem is that we do not need to consider about the boundary at $p_i=0$ or $q_i=0$.

We first consider the case of $\sigma_P > 0$ and $\sigma_Q > 0$, then $|w_i| ,|z_i|< 1$ for $1\leq i\leq n$.
Hence, the global minimum points must be stationary points, or be on the boundary at $\max_i |u^*_i| = R$.
By rearranging the order of $\{w_i\}$ appropriately, we can write $w^*_i > 0$ for $1 \leq i \leq I$ and $w^*_i = 0$ for $I+1 \leq i \leq n$.
The Lagrangian for the minimization problem \eqref{eq_min_problem}-\eqref{eq_constraint3} is given by 
\begin{align}
L(\bm{w}, \bm{z}, \bm{u}, \boldsymbol{\lambda})&:= - \sum_i w_i z_i + \sum_{k=1}^3 \lambda_k g_k(\bm{w}, \bm{u}) + \sum_{k=4}^6 \lambda_k g_k(\bm{z}, \bm{u}) \\
&=- \sum_i w_i z_i + \frac12 \sum_i w_i^2\phi_{\boldsymbol{\lambda}}(u_i) + \frac12\sum_i z_i^2\psi_{\boldsymbol{\lambda}}(u_i)-\sum_{k=1}^3 \lambda_k A_k-\sum_{k=1}^3 \lambda_{k+3}B_k,
\end{align}
where $\phi_{\boldsymbol{\lambda}}(u):= 2\sum_{k=1}^3 \lambda_k u^{k-1}$ and  $\psi_{\boldsymbol{\lambda}}(u):= 2\sum_{k=1}^3 \lambda_{k+3} u^{k-1}$.
Since $u_i\neq u_j$ for $i\neq j$, and $I\geq 3$, it follows that $\{\nabla g_k\}_{k\leq 6}$ are linearly independent.
Hence, if $\max_i|u^*_i| < R$, the global minimum points must satisfy
\begin{align}
\label{grad_1}
\frac{\partial{L}}{\partial{w_i}}&= -z^*_i+ w^*_i\phi_{\boldsymbol{\lambda}^*}(u^*_i)=0, \\[0.1cm]
\label{grad_2}
\frac{\partial{L}}{\partial{z_i}}&= -w^*_i+ z^*_i\psi_{\boldsymbol{\lambda}^*}(u^*_i)=0, \\[0.1cm]
\label{grad_3}
\frac{\partial{L}}{\partial{u_i}}&= \frac{1}{2}\Bigl({w^*_i}^2\phi'_{\boldsymbol{\lambda}^*}(u^*_i) +  {z^*_i}^2\psi'_{\boldsymbol{\lambda}^*}(u^*_i)\Bigr)=0,
\end{align}
where $'$ denotes the derivative with respect to $u$.
From $w^*_i > 0$ for $i\leq I$ and \eqref{grad_2}, it follows that $z^*_i \neq 0$, and from $w^*_i = 0$ for $i \geq I+1$ and \eqref{grad_1}, it follows that $z^*_i = 0$.
Thus, we obtain  
\begin{align}
\label{eq_pair}
(P^*, Q^*)\in \set{P}_{I,R}.
\end{align}
For $i\leq I$, from \eqref{grad_1} and \eqref{grad_2}, it follows that
\begin{align}
\label{eq_poly1}
\phi_{\boldsymbol{\lambda}^*}(u^*_i)\psi_{\boldsymbol{\lambda}^*}(u^*_i)-1=0.
\end{align}
By multiplying \eqref{grad_1} by $w^*_i $ and \eqref{grad_2} by $z^*_i$, and subtract each other, it follows that ${w^*_i}^2\phi_{\boldsymbol{\lambda}^*}(u^*_i)\ - {z^*_i}^2\psi_{\boldsymbol{\lambda}^*}(u^*_i) = 0$. By substituting \eqref{grad_3} into this relation, we have 
\begin{align}
\label{eq_poly2}
\phi_{\boldsymbol{\lambda}^*}(u^*_i)\psi'_{\boldsymbol{\lambda}^*}(u^*_i)+\phi'_{\boldsymbol{\lambda}^*}(u^*_i)\psi_{\boldsymbol{\lambda}^*}(u^*_i)=
\Bigl(\phi_{\boldsymbol{\lambda}^*}(u^*_i)\psi_{\boldsymbol{\lambda}^*}(u^*_i)-1\Bigr)'=0.
\end{align}
From \eqref{eq_poly1} and \eqref{eq_poly2}, the algebraic equation $\phi_{\boldsymbol{\lambda}^*}(u)\psi_{\boldsymbol{\lambda}^*}(u)-1=0$ has multiple roots at $\{u^*_i\}_{i\leq I}$, then the degree of this polynomial is greater than or equal to $2I$. 
If $\phi_{\boldsymbol{\lambda}^*}(u)\psi_{\boldsymbol{\lambda}^*}(u)-1$ is not identically zero, since $\phi_{\boldsymbol{\lambda}^*}(u)\psi_{\boldsymbol{\lambda}^*}(u)-1$ is a polynomial of degree at most $4$ in $u$, it must be $I\leq 2$.
If $\phi_{\boldsymbol{\lambda}^*}(u)\psi_{\boldsymbol{\lambda}^*}(u)-1$ is identically zero, from \eqref{grad_1} and $w^*_i=z^*_i=0$ for $i\geq I+1$, it follows that $\bm{w}^*=\bm{z}^*$.
It contradicts the assumption of  $m_P\neq m_Q$.
By combining $I\leq 2$ and \eqref{eq_pair}, if $\max_i|u^*_i| < R$, we obtain $(P^*,Q^*)\in\set{P}_2$.

We next consider the case of $\sigma_P > 0$ and $\sigma_Q=0$.
In this case, we can put $z_1=1$, $u_1=m_Q$, and $z_i=0$ for $i\geq 2$, then the Lagrangian is given by
\begin{align}
L(\bm{w}, \bm{u}, \boldsymbol{\lambda}):= -w_1 +\frac12 \sum_i w_i^2\phi_{\boldsymbol{\lambda}}(u_i) - \sum_{k=1}^3 \lambda_k A_k.
\end{align}
If $\max_i|u^*_i| < R$, the global minimum points must satisfy
\begin{align}
\label{grad_4}
\frac{\partial{L}}{\partial{w_1}}&= -1 + w^*_1\phi_{\boldsymbol{\lambda^*}}(m_Q)=0, \\[0.1cm]
\label{grad_5}
\frac{\partial{L}}{\partial{w_i}}&= w^*_i\phi_{\boldsymbol{\lambda^*}}(u^*_i)=0, \\[0.1cm]
\label{grad_6}
\frac{\partial{L}}{\partial{u_i}}&= \frac{1}{2}{w^*_i}^2\phi'_{\boldsymbol{\lambda^*}}(u^*_i)=0, \quad \mbox{for} \hspace*{0.15cm}  i\geq 2.
\end{align}
From \eqref{grad_4}, it follows that $w_1 \neq 0$, and $\phi_{\boldsymbol{\lambda^*}}(u)$ is not identically zero.
For $2\leq i \leq I$, we have $\phi_{\boldsymbol{\lambda^*}}(u^*_i)~=\phi'_{\boldsymbol{\lambda^*}}(u^*_i)=0$.
Hence, the algebraic equation $\phi_{\boldsymbol{\lambda}^*}(u)=0$ has multiple roots at $\{u^*_i\}_{2\leq i\leq I}$, then we have $I\leq 2$ from $1\leq \deg \phi_{\boldsymbol{\lambda}^*} \leq 2$.
The proof for the case of $\sigma_P = 0$ and $\sigma_Q > 0$ is the same.  By combining cases of $\sigma_P=\sigma_Q=0$, if $\max_i |u_i| < R$, we obtain $(P^*,Q^*)\in\set{P}_2$.
\end{proof}

\begin{proof}[Proof of Theorem~\ref{theorem_LB_Hellinger}]
We first prove \eqref{eq_LB_Hellinger} for pairs of finite discrete probability measures. \\
Let $H^*:= \inf_{(P,Q)\in\set{P}_n \cap \set{P}[m_P, \sigma_P, m_Q, \sigma_Q]} H^2(P,Q)$ and suppose $H^*<h^2(r,s)$.
By applying Lemma~\ref{lem_minimum} as $R\rightarrow\infty$, there exist sequences of vectors $\{\bm{u}_j\}$, and probability measures $\{P_j\}$ and $\{Q_j\}$, which are defined on $\{\bm{u}_j\}$, such that
\begin{align}
H^2(P_\infty,Q_\infty)=H^*,
\end{align}
where $Z_\infty$ denotes $\lim_{j\rightarrow \infty} Z_j$ for $Z=\{P,Q, u_i\}$. Without any loss of generality, one can assume that $|u_{i,\infty}| < \infty$ for $1\leq i\leq I$ and $|u_{i,\infty}| = \infty$ for $I+1\leq i\leq n$. Let $\sum_{i\geq I+1} p_{i,\infty}u^2_{i,\infty}=C^2$ and  $\sum_{i\geq I+1} q_{i,\infty}u^2_{i,\infty}=D^2$, where $p_{i,j}=P_j(u_{i,j})$ and $q_{i,j}=Q_j(u_{i,j})$. From the variance constraints, we have $C^2\leq m_P^2+\sigma_P^2$ and $D^2\leq m_Q^2+\sigma_Q^2$. Hence, $p_{i,\infty}=O(u_{i,\infty}^{-2})$ and $q_{i,\infty}=O(u_{i,\infty}^{-2})$ for $i \geq I+1$, then
\begin{align}
\label{eq_moment}
\sum_{i\geq I+1} p_{i,\infty}=\sum_{i\geq I+1} p_{i,\infty}u_{i,\infty}=\sum_{i\geq I+1} \sqrt{p_{i,\infty}q_{i,\infty}}=0.
\end{align}
Let $P'$ and $Q'$ be probability measures defined on $\{u_{1,\infty}, u_{2,\infty}, \cdots, u_{I,\infty}\}$, and let $P'(u_{i,\infty})=p_{i,\infty}$, $Q'(u_{i,\infty})=q_{i,\infty}$ for $1\leq i \leq I$.
From \eqref{eq_moment}, it follows that 
\begin{align}
(P',Q')&\in \set{P}_{I,R} \cap \set{P}[m_P, \sigma_P^2 -C^2, m_Q, \sigma_Q^2 -D^2], \\
H^*&=H^2(P', Q'),
\end{align}
where $R > \max_{i\leq I} |u_{i, \infty}|$. Since the variances of $P'$ and $Q'$ are non-negative, we have $0 \leq C^2\leq \sigma_P^2$, and $0 \leq D^2\leq \sigma_Q^2$. 
By considering the similar sequences under constraints $\sum_{i\geq I+1} p_{i,\infty}u^2_{i,\infty}=C^2$ and $\sum_{i\geq I+1} q_{i,\infty}u^2_{i,\infty}=D^2$, and using the definition of $H^*$, it follows that $H^*$ is the global minimum in $\set{P}_{I,R} \cap \set{P}[m_P, \sigma_P^2 -C^2, m_Q, \sigma_Q^2 -D^2]$. Hence, by applying Lemma~\ref{lem_minimum} for sufficiently large $R$, it follows that $(P',Q')\in \set{P}_2$, and
\begin{align}
H^*=H^2(P', Q')= h^2(r',s'),
\end{align}
where $r'$ and $s'$ satisfy the moment constraints of $(m_P,  \sigma_P^2 -C^2, m_Q, \sigma_Q^2 -D^2)$, and they are unique from Lemma~\ref{lem_p2}.
From Lemma \ref{lem_binary}, since $h^2(r,s)$ is monotonically decreasing with respect to $\sigma_P$ and $\sigma_Q$, we have $H^* = h^2(r',s') \geq h^2(r,s)$. This contradicts the assumption of $H^*<h^2(r,s)$, then we obtain $H^*=h^2(r,s)$.

Next, we prove \eqref{eq_LB_Hellinger} for pairs of probability measures in $\set{P}[m_P,\sigma_P, m_Q, \sigma_Q]$. For an arbitrary small $\epsilon$, there exists $R$ such that
\begin{align}
\label{approximate1}
|\int_{|x|>R} px^{k-1} \mathrm{d}\mu(x)|<\epsilon, \quad |\int_{|x|>R} qx^{k-1} \mathrm{d}\mu(x)|<\epsilon, \quad \mbox{for} \hspace*{0.15cm}  k=1,2,3.
\end{align}
Since $(\sqrt{p}-\sqrt{q})^2\leq p+q$, we have 
\begin{align}
\label{approximate2}
\frac12 \int_{|x|>R} (\sqrt{p}-\sqrt{q})^2 \mathrm{d}\mu(x) < \epsilon.
\end{align}
In the interval $[-R, R]$, one can approximate probability measures by finite discrete probability measures $(P_d, Q_d)\in\set{P}_{n, R}$ as follows.
\begin{align}
\label{approximate3}
&|\int_{|x|\leq R} px^{k-1} \mathrm{d}\mu(x)- \sum_i p_i u_i^{k-1}|<\epsilon, \\[0.1cm]
\label{approximate4}
&|\int_{|x|\leq R} qx^{k-1} \mathrm{d}\mu(x)- \sum_i q_i u_i^{k-1}|<\epsilon, \quad \mbox{for} \hspace*{0.15cm}  k=1,2,3, \\[0.1cm]
\label{approximate5}
&|\frac12 \int_{|x|\leq R}  (\sqrt{p}-\sqrt{q})^2 \mathrm{d}\mu(x)- H^2(P_d, Q_d)|<\epsilon.
\end{align}
From \eqref{approximate2} and \eqref{approximate5}, we have $H^2(P,Q)=H^2(P_d, Q_d)+O(\epsilon)$, and from \eqref{approximate1}, \eqref{approximate3}, and \eqref{approximate4}, differences of means and variances between $(P,Q)$ and $(P_d, Q_d)$ are $O(\epsilon)$. 
By applying $H^2(P_d, Q_d)\geq h^2(r_d, s_d)$, it follows that
\begin{align}
H^2(P,Q)=H^2(P_d, Q_d)+O(\epsilon)\geq h^2(r_d, s_d)+O(\epsilon)=h^2(r, s)+O(\epsilon),
\end{align}
where $(r_d, s_d)$ satisfy the momentum constraints of $(m_{P_d}, \sigma_{P_d}, m_{Q_d}, \sigma_{Q_d})$, and we use differentiability of $h^2(r,s)$ with respect to moments and variances.
Since $\epsilon$ is arbitrary small, we obtain $H^2(P,Q) \geq h^2(r,s)$.

We finally show Item(c). The proof is the same as Theorem 2 in~\cite{nishiyama2020relations}, so we outline of the proof.  We construct sequence of probability measures $\{(P_j, Q_j)\}$ with zero mean and respective variances $(\sigma_P^2, \sigma_Q^2)$ for which $H^2(P_j, Q_j)\rightarrow 0$ as $j\rightarrow \infty$
(without any loss of generality, one can assume that the equal means are equal to zero).
We start by assuming $\min\{\sigma_P^2, \sigma_Q^2\}\geq 1$.
Let 
\begin{align}
\mu_j:= \sqrt{1 + j(\sigma_Q^2-1)},
\end{align} 
and define a sequence of quaternary real-valued random variables with probability mass functions
\begin{align}
Q_j(a) := 
\begin{dcases}
\frac12-\frac{1}{2j}, & \quad a = \pm 1, \\
\frac{1}{2j},       & \quad a = \pm \mu_j.
\end{dcases}
\end{align}
It can be verified that, for all $j\in \naturals$, $Q_j$ has zero mean and variance $\sigma_Q^2$.

Furthermore, let
\begin{align}
P_j(a) := 
\begin{dcases}
\frac{1}{2}-\frac{\xi}{2j}, & \quad a = \pm 1, \\
\frac{\xi}{2j},       & \quad a = \pm \mu_j,
\end{dcases}
\end{align}
with
\begin{align}
\xi:= \frac{\sigma_P^2-1}{\sigma_Q^2-1}.
\end{align}
If $\xi > 1$, for $j=1, \cdots, \lceil \xi \rceil$, we choose $P_j$ arbitrary with mean $0$ and variance $\sigma_P^2$.
Then,
\begin{align}
H^2(P_j, Q_j)=h^2\Bigl(\frac{\xi}{j}, \frac{1}{j}\Bigr)\rightarrow 0.
\end{align}
Next, suppose $\min\{\sigma_P^2, \sigma_Q^2\}:= \sigma^2<1$, then construct $P'_j$ and $Q'_j$ as before with variances $\frac{2\sigma_P^2}{\sigma^2}>1$ and $\frac{2\sigma_Q^2}{\sigma^2}>1$, respectively. If $P_j$ and $Q_j$ denote the random variables $P'_j$ and $Q'_j$ scaled by a factor of $\frac{\sigma}{\sqrt{2}}$, then their variances are $\sigma_P^2, \sigma_Q^2$, respectively, and $H^2(P_j,Q_j)=H^2(P'_j,Q'_j)\rightarrow 0$ as we let $j\rightarrow \infty$.
\end{proof}

\begin{proof}[Proof of Proposition~\ref{proposition: hellinger_LB}]
We prove the first inequality because the proof for the second inequality is the same. From \eqref{eq-binary Hellinger}, we obtain
\begin{align}
\label{2020c1}
h^2(r,s) &= \frac{(r-s)^2}{2}\Bigl(\frac{1}{(\sqrt{r}+\sqrt{s})^2} + \frac{1}{(\sqrt{1-r}+\sqrt{1-s})^2}\Bigr) \\[0.1cm]
\label{2020c2}
& = \frac{(r-s)^2}{2}\Bigl(g\Bigl(\frac{s}{r}\Bigr)\frac{1}{r+s} + g\Bigl(\frac{1-s}{1-r}\Bigr)\frac{1}{2-r-s} \Bigr) \\[0.1cm]
\label{2020c3}
& \geq \beta_{\mathrm{min}}\frac{(r-s)^2}{4}\Bigl(\frac{1}{r+s} + \frac{1}{2-r-s} \Bigr).
\end{align}
From \eqref{r} and \eqref{s}, it follows that
\begin{align}
\label{2020c4}
r+s=1 + \frac{b}{2av}, \\[0.1cm]
\label{2020c5}
r-s=\frac{a}{2v}.
\end{align}
Substituting \eqref{2020c4} and \eqref{2020c5} into \eqref{2020c3}, it follows that
\begin{align}
h^2(r,s) &\geq \beta_{\mathrm{min}}\frac{a^2}{16v^2}\frac{2}{1-\frac{b^2}{4a^2v^2}} \\[0.1cm]
&=\beta_{\mathrm{min}}\frac{a^2}{2\Bigl(a^2+2(\sigma_P^2+\sigma_Q^2)\Bigr)}=\beta_{\mathrm{min}}l(m_P,\sigma_P, m_Q, \sigma_Q).
\end{align}
Hence, we obtain \eqref{eq_compare}.
\end{proof}

\section{Open problems} 

As shown in this note, the squared Hellinger distance, chi-squared divergence, and relative entropy satisfy the following condition.\\
\textbf{Lower bound condition} \\
{\em The binary divergence attains the lower bound with given means and variances.}\\

These divergences belong to $f$-divergence and the $\alpha$-divergence \cite{cichocki2010families} given by
\begin{align}
D_A^{(\alpha)}(P\|Q):= \Bigl(\frac{1}{\alpha(\alpha-1)}\int p^\alpha q^{1-\alpha} \mathrm{d}\mu -1\Bigr).
\end{align}
We conclude with open problems.
\begin{problem}
What are conditions for a function $f$ of the $f$-divergence to satisfy the ``Lower bound condition''?
\end{problem}
As a simpler problem than Problem 1, 
\begin{problem}
What are conditions of $\alpha$ to satisfy the ``Lower bound condition''?
\end{problem}

\bibliography{reference_Hellinger} 
\bibliographystyle{myplain}

\end{document}